%% file: paper.tex
\documentclass[conference]{IEEEtran}
\IEEEoverridecommandlockouts

\usepackage{cite}
\usepackage{amsmath,amssymb,amsfonts,amsthm}
\usepackage{algorithmic}
\usepackage{graphicx}
\usepackage{textcomp}
\usepackage{xcolor}

\usepackage[algo2e,linesnumbered,ruled,lined,noend]{algorithm2e}
\usepackage{booktabs}
\usepackage[table]{xcolor}
\usepackage[hidelinks]{hyperref}
\usepackage{multirow}
\usepackage{multicol}
\usepackage{enumitem}
\usepackage[caption=false,font=footnotesize,justification=centering]{subfig}
\usepackage{centernot}
\usepackage{makecell}
\usepackage{adjustbox}
\usepackage{bigstrut}
\usepackage{ulem}

\setlist[itemize]{topsep=0pt}

\SetCommentSty{mycommfont}

\def\BibTeX{{\rm B\kern-.05em{\sc i\kern-.025em b}\kern-.08em
    T\kern-.1667em\lower.7ex\hbox{E}\kern-.125emX}}
\begin{document}
\input{100config}

\title{
HyperSearch: Prediction of New Hyperedges through Unconstrained yet Efficient Search
}

\author{
Hyunjin Choo\textsuperscript{1},
Fanchen Bu\textsuperscript{1},
Hyunjin Hwang\textsuperscript{2}
Young-Gyu Yoon\textsuperscript{1},
Kijung Shin\textsuperscript{1,2}\\
\textsuperscript{1}School of Electrical Engineering and
\textsuperscript{2}Kim Jaechul Graduate School of AI, KAIST, Republic of Korea\\
\{choo, boqvezen97, hyunjinhwang, ygyoon, kijungs\}@kaist.ac.kr
}

\maketitle

\begin{abstract}
\input{000abstract}
\end{abstract}


\input{010intro}
\input{020related}

\input{030preliminery}

\input{040observation}
\input{050proposed}
\input{060experiment}
\input{070conclusion}

\section*{Acknowledgment}
{\small This work was partly supported by Institute of Information \& Communications Technology Planning \& Evaluation (IITP) grant funded by the Korea government (MSIT) 
(No. RS-2022-II220157, Robust, Fair, Extensible Data-Centric Continual Learning)
(No. RS-2024-00457882, AI Research Hub Project)
(RS-2019-II190075, Artificial Intelligence Graduate School Program (KAIST)).
This research was partly supported by National Research Foundation of Korea (NRF) (RS-2023-00209473), Korea Basic Science Institute (National research Facilities and Equipment Center) grant funded by the Ministry of Science and ICT (RS-2024-00401676).}

\vspace{-1pt}
\normalem
\bibliographystyle{IEEEtran}
\bibliography{080references}



\end{document}

%% file: 100config.tex
\newcommand{\smallsectiontwo}[1]{{\bf{\underline{\smash{#1}}}}}
\newcommand\red[1]{\textcolor{red}{#1}}
\newcommand\blue[1]{\textcolor{blue}{#1}}
\newcommand\cyan[1]{\textcolor{cyan}{#1}}
\newcommand\violet[1]{\textcolor{violet}{#1}}
\newcommand\teal[1]{\textcolor{teal}{#1}}
\newcommand\brown[1]{\textcolor{brown}{#1}}

\definecolor{darkgreen}{RGB}{0, 150, 0}   
\definecolor{goldenrod}{RGB}{218, 165, 32} 

\newcommand\green[1]{\textcolor{darkgreen}{#1}}
\newcommand\yellow[1]{\textcolor{goldenrod}{#1}}

\newcommand\hwang[1]{\textcolor{blue}{#1}}

\newcommand\fanchen[1]{\textcolor{blue}{[Fanchen: #1]}}

\newtheorem{theorem}{Theorem}
\newtheorem{lemma}{Lemma}

\newtheorem{obs}{Observation}
\newtheorem{defn}{Definition}
\newtheorem{problem}{Problem}

\newcommand{\smallsection}[1]{{\noindent {\bf{\underline{\smash{#1}}}}}}

\newcommand{\methodthree}{HyperSearch\xspace}

\setlength{\textfloatsep}{0.12cm}
\setlength{\dbltextfloatsep}{0.12cm}
\setlength{\abovecaptionskip}{0.12cm}
\setlength{\skip\footins}{0.15cm}

%% file: 000abstract.tex
Higher-order interactions (HOIs) in complex systems, such as scientific collaborations, multi-protein complexes, and multi-user communications, are commonly modeled as hypergraphs, where each hyperedge (i.e., a subset of nodes) represents an HOI among the nodes.
Given a hypergraph, \textit{hyperedge prediction} aims to identify hyperedges that are either missing or likely to form in the future, and it has broad applications, including recommending interest-based social groups, predicting collaborations, and uncovering functional complexes in biological systems.
However, the vast search space of hyperedge candidates (i.e., all possible subsets of nodes) poses a significant computational challenge, making na\"ive exhaustive search infeasible.
As a result, existing approaches rely on either heuristic sampling to obtain constrained candidate sets or ungrounded assumptions on hypergraph structure to select promising hyperedges.

In this work, we propose \methodthree, a search-based algorithm for hyperedge prediction that \textit{efficiently} evaluates \textit{unconstrained} candidate sets, 
by incorporating two key components:
(1) an \textit{empirically grounded} scoring function derived from observations in real-world hypergraphs and 
(2) an \textit{efficient} search mechanism, where we derive and use an anti-monotonic upper bound of the original scoring function (which is not anti-monotonic) to prune the search space.
This pruning comes with \textit{theoretical guarantees}, ensuring that discarded candidates are never better than the kept ones w.r.t. the original scoring function.
In extensive experiments on 10 real-world hypergraphs across five domains, \methodthree consistently outperforms state-of-the-art baselines, achieving higher accuracy in predicting new (i.e., not in the training set) hyperedges.

%% file: 010intro.tex
\section{Introduction}
\label{ch3sec:intro}
    In many complex real-world systems, groups of entities engage simultaneously, forming \textit{higher-order interactions} (HOIs), e.g., scientific collaborations (coauthorship among multiple researchers)~\cite{benson2018simplicial}, multi-protein complexes in biological systems~\cite{klimm2021hypergraphs}, and multi-user communications~\cite{iacopini2022group}.

    \textit{Hypergraphs}, which extend ordinary (pairwise) graphs by allowing hyperedges to connect multiple nodes, offer a natural framework for modeling HOIs.
    Unlike graphs, where an edge links exactly two nodes, hyperedges capture arbitrary-sized relationships, providing a more expressive representation of complex multi-way interactions.
    This enables the discovery of structural patterns and behaviors that are often obscured in pairwise frameworks~\cite{benson2018simplicial,lee2024survey,juul2024hypergraph}.


    Given a hypergraph, \textit{hyperedge prediction} 
    aims to identify hyperedges that are missing or to emerge in the future based on observed data.
    This task is essential for improving predictive capabilities in systems where HOIs play a crucial role.
    Hyperedge prediction has broader applications, including:
    
    \begin{itemize}[leftmargin=*,topsep=0pt]
        \item \textbf{Group recommendation}: In social networks, hyperedge prediction enables the recommendation of groups, each consisting of users sharing common interests or behaviors,
        {enhancing user experience in social media platforms and improving the effectiveness of targeted marketing}~\cite{liben2003link, wang2014link}.
        \item \textbf{Collaboration prediction}: In academic and industrial networks, hyperedge prediction helps identify potential collaborations among researchers or companies with shared interests or expertise, optimizing team formation 
        \cite{wang2014heterogeneous, lande2020link}.
        \item \textbf{Drug discovery}: In protein and gene networks, hyperedge prediction aids in forecasting functional protein complexes or interacting gene groups, providing insights into disease mechanisms and facilitating drug discovery~\cite{jin2023general, saifuddin2023hygnn}.
    \end{itemize}
    
    However, 
    the problem of hyperedge prediction is computationally challenging. The number of potential hyperedges is $\mathcal{O}(2^{n})$ for $n$ nodes, making exhaustive enumeration infeasible for even medium-sized hypergraphs with hundreds of nodes.

    While various methods have been proposed to efficiently identify promising hyperedges \cite{kumar2020hpra,zhang2019hyper,yu2024mhp,hwang2022ahp,yadati2020nhp},
    they
    suffer from two fundamental limitations:
    \begin{itemize}[leftmargin=*,topsep=0pt]
        \item \textbf{{Constrained} candidate sets}: Most deep learning-based methods (e.g., Hyper-SAGNN~\cite{zhang2019hyper} and NHP~\cite{yadati2020nhp}) frame hyperedge prediction as a binary classification task, distinguishing between existing (positive) and nonexistent (negative) hyperedges.
        Since enumerating all potential hyperedges is computationally infeasible, negative hyperedges are typically \textit{sampled heuristically}.
        Consequently, the performance of such methods strongly depends on the quality of these samples, {and how to sample desirable negative hyperedges is a challenging problem itself (see, e.g., AHP ~\cite{hwang2022ahp}).}
        Alternatively, MHP~\cite{yu2024mhp} {assumes a given query node set $Q$ and predicts hyperedges by initializing hyperedges with nodes in $Q$ and filling the missing nodes.}
        However, the predictions by MHP are \textit{constrained by the choice} of the query node set $Q$, and selecting $Q$ is non-trivial in practice, {limiting its flexibility and generalizability}.
        \item \textbf{{Unjustified structural assumptions}}:
        HPRA~\cite{kumar2020hpra} constructs candidate hyperedges by adding ``similar'' nodes based on resource allocation scores. 
        However, these scores \textit{strongly depends on observed structures} and assume specific \textit{local connectivity patterns} {(e.g., the similarity between two nodes is proportional to the number of neighbors they share), \textit{without sufficient justification}}. 
        Such assumptions make the method less flexible and generalizable, particularly when discovering new hyperedges connecting nodes with low structural similarity.
        Moreover, HPRA is limited to predicting hyperedges consisting only of nodes within the same connected component, as the similarity between nodes in different components is zero. 
        This constraint limits its applicability in sparse or disconnected hypergraphs.
    \end{itemize}   

    To address these limitations, we propose \methodthree, a search-based algorithm for hyperedge prediction.
    \methodthree is
    (1) \textit{empirically justified}, i.e., grounded in real-world observations rather than ad-hoc heuristics, enabling generalization beyond observed structures, and 
    (2) \textit{unconstrained yet efficient}, i.e., efficiently explores the vast hyperedge search space without being constrained by predefined or sampled subsets.
    \begin{itemize}[leftmargin=*,topsep=0pt]
        \item \textbf{Empirically justified scores based on observations}:
        The scoring function in \methodthree is designed based on our observations in real-world datasets.
        Specifically, we observe that ground-truth hyperedges have \textit{significant overlap} with observed hyperedges, and we thus assign each candidate hyperedge a score based on the number of observed hyperedges it significantly overlaps with.
        Additionally, when node features or timestamps are available, we incorporate node feature similarity to prioritize candidates consisting of nodes with similar features, and apply time weighting to emphasize recent hyperedges when computing overlaps, based on our observations.
        \item \textbf{Efficient search with an anti-monotonic upper bound}:
        The original scoring function is not anti-monotonic, i.e., a subset may have a lower score than its supersets, rendering naïve pruning strategies ineffective. 
        To address this challenge, we derive an anti-monotonic upper bound of the original scoring function.
        This allows us to efficiently explore the search space with theoretical guarantees, ensuring that discarded candidates are never better than
        the kept ones w.r.t. the original scoring function.
    \end{itemize}


    We conduct extensive experiments on 10 real-world hypergraphs across five domains, which demonstrate the empirical superiority of \methodthree:    
    \begin{itemize}[leftmargin=*,topsep=0pt]
        \item \textbf{Predictive accuracy}: \methodthree consistently outperforms all baseline methods, including deep learning-based and rule-based ones, across diverse datasets, achieving higher accuracy in predicting new hyperedges.
        \item \textbf{Computational efficiency}: \methodthree achieves faster runtime than deep learning-based methods in most cases, while maintaining near-linear scalability with respect to the input hypergraph size. 
        \item \textbf{Component effectiveness}: Ablation studies empirically validate that each component of \methodthree contributes meaningfully to performance.
    \end{itemize}


\smallsection{Reproducibility:}
The code, datasets, and appendix 
are available at \url{https://github.com/jin-choo/HyperSearch}.


    The remainder of this paper is organized as follows.
    In Sect.~\ref{ch3sec:related}, we review related work.
    In Sect.~\ref{ch3sec:prelim}, we introduce preliminary concepts and describe the datasets.
    In Sect.~\ref{ch3sec:observation}, we present empirical observations from real-world datasets.
    In Sect.~\ref{ch3sec:proposed}, we detail our proposed method.
    In Sect.~\ref{ch3sec:experiments}, we review our experiments. 
    In Sect.~\ref{ch3sec:conclusion}, we conclude the paper.

%% file: 020related.tex
\section{Related Work}
\label{ch3sec:related}


\subsection{Similarity-Based Methods}
Early approaches to hyperedge prediction are mainly based on node similarity measures~\cite{zhang2018beyond, yoon2020much}.
They either extend node similarity measures to hypergraphs or project hypergraphs into pairwise graphs to apply node similarity measures.
HPRA~\cite{kumar2020hpra}, for example, constructs hyperedges by selecting a seed node via preferential attachment (i.e., high-degree nodes are more likely to be chosen) and incrementally adding structurally similar nodes, guided by resource allocation scores.

\smallsection{Limitations.}
While such methods are computationally efficient, they primarily depend on local structural similarity (e.g., pairwise similarity) and are limited in capturing the complex higher-order dependencies inherent in real-world hypergraphs.


\subsection{Deep Learning-Based Methods}
To overcome the limitations of similarity-based approaches, deep learning models have been employed to learn higher-order representations in hypergraphs.
DHNE~\cite{tu2018structural} proposes an MLP-based framework for $k$-uniform hypergraphs.
Hyper-SAGNN~\cite{zhang2019hyper} extends this approach by incorporating a self-attention mechanism with graph neural networks, enabling non-$k$-uniform hyperedge prediction.
NHP~\cite{yadati2020nhp} further generalizes hyperedge prediction to directed hypergraphs and extends it to inductive settings.

\smallsection{Limitations.}
While these methods demonstrate strong performance, they typically rely on predefined candidate sets and require negative sampling strategies for effective training, as described below.

\subsection{Negative Sampling on Hyperedges}
Hyperedge prediction is commonly framed as a binary classification task, where the model distinguishes between observed (positive) and nonexistent (negative) hyperedges.
Given the combinatorial explosion of potential hyperedges,\footnote{The number of potential hyperedges is $\mathcal{O}(2^{n})$ for $n$ nodes.} negative sampling plays a crucial role in defining meaningful training samples, and different negative sampling strategies introduce varying levels of difficulty, directly impacting prediction performance~\cite{patil2020negative}.

To address this challenge, AHP~\cite{hwang2022ahp} introduces a generative adversarial training approach that dynamically generates negative samples during training, mitigating the dependence on heuristics.
In contrast, MHP~\cite{yu2024mhp} proposes a novel masking-based approach that eliminates the need for explicit negative sampling by formulating hyperedge prediction as a masked node prediction task given a query node set and target size.

\smallsection{Limitations.}
However, they still depend on external constraints, such as the choice of query set in MHP, that inherently limit the candidate space.
Consequently, 
their ability to explore hyperedge candidates remains limited, especially when compared to approaches like ours that search the full space without predefined candidates.

    \begin{figure*}
        \centering
    
        \begin{minipage}[c]{.595\textwidth}
            \centering
            \includegraphics[width=0.542\textwidth]{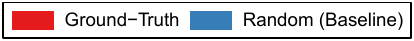}
        \end{minipage}
        \hspace{15pt}
        \begin{minipage}[c]{.325\textwidth}
            \centering
            \includegraphics[width=0.85\textwidth]{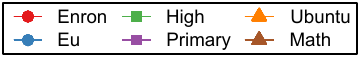}
		\end{minipage}
    
        \vspace{-8pt}
    
        \subfloat[\textbf{Obs. 1. \uline{There is significant structural overlap} \\ \uline{between new hyperedges and observed hyperedges.}}\label{ch3fig:observation1}]{
            \includegraphics[width=0.595\textwidth]{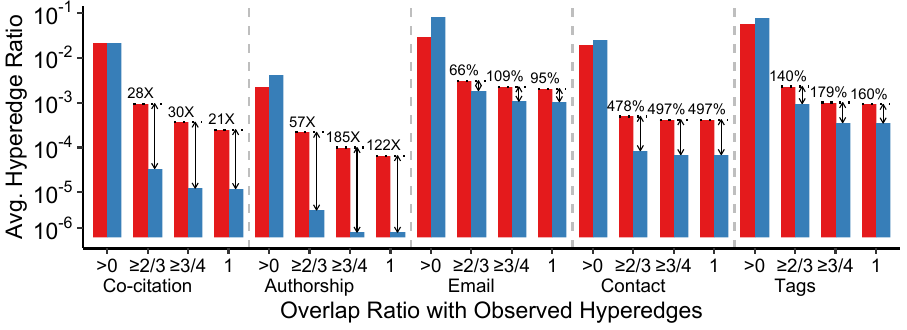}
        }
        \hspace{13pt}
        \subfloat[\textbf{Obs. 2. \uline{Overlap increases as} \\ \uline{the time difference decreases.}}\label{ch3fig:obs_time}]{
            \includegraphics[width=0.325\textwidth]{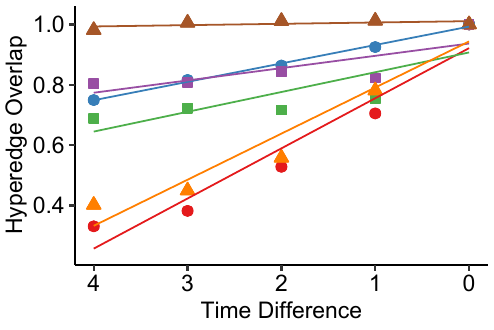}
        }    
        \caption[\textbf{Observations.}]{
            \textbf{Observations.}
            (a) The average proportion of ground-truth new hyperedges with overlap ratios with low threshold ($>0$) is comparable to that of the random ones.
            However, for high thresholds ($\geq 2/3$, $\geq 3/4$, and $1$), the gaps between ground-truth and random ones become significant, suggesting \uline{ground-truth new hyperedges are more likely to substantially overlap with existing hyperedges than random hyperedges are}.
            (b) Structural overlap between hyperedges declines as the time gap increases, suggesting \uline{new hyperedges tend to overlap more substantially with recent existing ones than with earlier ones}.
        }
        \label{ch3fig:observations}
    \end{figure*}

%% file: 030preliminery.tex
\section{Preliminaries and Datasets}
\label{ch3sec:prelim}

\subsection{Concepts and Problem Statement}
\label{ch3sec:prelim:concepts}
    \subsubsection{Hypergraphs}
        A \textit{hypergraph} $H=(V, E)$ is defined by its \textit{node set} $V=\{ v_1, v_2, \ldots, v_{|V|} \}$ and \textit{hyperedge set} $E=\{ e_1, e_2, \ldots, e_{|E|} \}$.
        Each hyperedge $e_j$ is a non-empty subset of nodes (i.e., $\emptyset \neq e_j \subseteq V$), and its \textit{edge size} $|e_j|$ is the number of nodes in it.
        Each hyperedge $e_j$ can also be associated with an \textit{edge timestamp} $t_j$, if such information is available.
        \subsubsection{Problem Definition (Hyperedge Prediction)} \label{ch3sec:problem}
        \hfill\\
        \noindent
        \fbox{
        \parbox{0.95\columnwidth}{
        \begin{itemize}[leftmargin=*]
        \item \textbf{Given:}
            \begin{enumerate}[leftmargin=*]
                \item a (partially) observed hypergraph $H=(V, E)$\\(optionally with 
                edge timestamps $t_j$'s),
                \item the target number $k$ of hyperedges to predict,
            \end{enumerate}
        \item \textbf{To find:} a set $E^{\prime}$ of $k$ predicted new hyperedges\\        (i.e., $E^{\prime} \subseteq 2^V \setminus (E \cup \{\emptyset\})$ with $|E^{\prime}| = k$),
        \item \textbf{Aim to:} 
        Accurately predict hyperedges that are missing or are to emerge in the future in $H$.
        \end{itemize}
        }
        }

\subsection{Datasets}
\label{ch3sec:prelim:data}
\noindent

We use 10 real-world hypergraphs from five different domains, summarized in Table~\ref{tab:datasets}.\footnote{Source 1: \url{https://www.cs.cornell.edu/~arb/data}\\\hspace*{8.8pt}
Source 2: \url{https://github.com/HyunjinHwn/SIGIR22-AHP}}
{Among them, the Email, Contact, and Tags datasets contain edge timestamps.}
\begin{itemize}[leftmargin=*,topsep=0pt]
\setlength\itemsep{0em}
    \item \textbf{Co-citation (Citeseer, Cora)}: Each hyperedge consists of publications (nodes) cited by the same paper.
    \item \textbf{Authorship (Cora-A, DBLP-A)}: Each hyperedge consists of publications (nodes) by the same author.
    \item \textbf{Email (Eu, Enron)}: Each hyperedge consists of the sender (node) and all receivers (nodes) of an email.
    \item \textbf{Contact (High, Primary)}: 
    Each hyperedge includes people (nodes) in a group interaction recorded by wearable sensors.
    \item \textbf{Tags (Math.sx, Ubuntu)}:
    Each hyperedge consists of the set of tags (nodes) added to the same question.
\end{itemize}

\begin{table}[t]
\caption{\textbf{Summary of real-world hypergraph datasets.}}\label{tab:datasets}
\centering
\small
\setlength{\tabcolsep}{8pt}
\renewcommand{\arraystretch}{0.92}
\begin{adjustbox}{max width=\linewidth}
\begin{tabular}{llcccc}
    \toprule
    \textbf{Domain} & \textbf{Dataset} & \textbf{\# Nodes} & \textbf{\# Hyperedges} & \textbf{Timestamps} \\ \midrule
    \multirow{2}{*}{Co-citation} & Citeseer & 1,457 & 1,078 & \\
    & Cora & 1,434 & 1,579 & \\ \midrule
    \multirow{2}{*}{Authorship} & Cora-A & 2,388 & 1,072 & \\
    & DBLP-A & 39,283 & 16,483 & \\ \midrule
    \multirow{2}{*}{Email} & Enron & 143 & 10,883 & \multirow{2}{*}{\checkmark} \\
    & Eu & 998 & 234,760 & \\ \midrule
    \multirow{2}{*}{Contact} & High & 327 & 172,035 & \multirow{2}{*}{\checkmark} \\
    & Primary & 242 & 106,879 & \\ \midrule
    \multirow{2}{*}{Tags} & Math.sx & 1,629 & 822,059 & \multirow{2}{*}{\checkmark} \\
    & Ubuntu & 3,029 & 271,233 & \\
    \bottomrule
\end{tabular}
\end{adjustbox}
\end{table}

%% file: 040observation.tex
\section{Observations}
\label{ch3sec:observation}
    This section highlights our key observations in real-world hypergraphs.
    We shall later design the scoring function in the proposed method (see Sect.~\ref{ch3sec:proposed}) based on these observations.

\subsection{Significant Overlap between Hyperedges}
\label{ch3obs:obs1}\label{ch3obs:obs2}
    Our first observation is about the \textit{overlap} between new hyperedges (which we aim to predict) and observed (training) hyperedges.
    We define the \textit{overlap ratio} between a new hyperedge $e'$ and an observed hyperedge $e$ as 
    $OverlapRatio(e^{\prime}, e) := {| e^{\prime} \cap e |} / {| e |}$.
    
    \smallsection{Setup:}
    We randomly split the hyperedge set $E$ in each dataset into two subsets: 
    observed hyperedges $E_1$ ($80\%$ of $E$) and (ground-truth) new hyperedges $E_2$ (remaining $20\%$ of $E$), using five random seeds.
    For each new hyperedge in $E_2$, we compute the proportion of observed hyperedges in $E_1$ with different overlap-ratio thresholds: $>0$, $\geq2/3$, $\geq3/4$, and $1$, and report the values averaged
    over all hyperedges in $E_2$.
    
    As a baseline for comparison, we computed the same statistics on random hyperedges generated by a generalized Chung-Lu (CL) model \cite{chung2002average}, which preserves both node-degree and hyperedge-size distributions of the ground-truth $E_2$.
    Specifically, in each hyperedge, nodes are replaced with randomly drawn nodes, where each node is selected independently with probability proportional to its degree.
    
    \smallsection{Observations:}
    As shown in Fig.~\ref{ch3fig:observations}(a), the proportions of overlapping observed hyperedges for low threshold ($>0$) show no significant difference between ground-truth hyperedges and the random ones.
    However, for higher thresholds ($\geq 2/3$, $\geq 3/4$, and $=1$), ground-truth hyperedges exhibit significantly more overlap with observed hyperedges than the random ones.
    This indicates that new hyperedges are more likely to substantially overlap with existing hyperedges than random hyperedges are, and this difference is statistically significant.~\footnote{
    At a higher threshold (spec., $\geq 2/3$), the Kolmogorov–Smirnov (KS) distance between the overlap-ratio distributions of ground-truth and random hyperedges is statistically significant ($KS = 0.26$, $p < 0.01$), confirming that the two distributions are clearly distinguishable with high confidence.}
\subsection{Temporal Bias in Structural Overlap}\label{ch3obs:obs3}

Our second observation is about the impact of edge timestamps on structural overlap between hyperedges.

\smallsection{Setup:}
We partition the hyperedges $E$ evenly into five equal-sized groups based on their timestamps, arranged in temporal order, where the first group $G_1$ contains the earliest hyperedges, the last group $G_5$ contains the most recent (latest) ones, etc.
For each pair of groups $G_i$ and $G_j$ with $i < j$, we see the hyperedges in $G_i$ as observed ones and the ones in $G_j$ as new ones, and compute the average overlap ratio (see Sect.~\ref{ch3obs:obs1}) between $e' \in G_j$ and $e \in G_i$.
We further average the ratios w.r.t. the time gap $j - i$ (larger $j - i$ means longer time gaps).

\smallsection{Observations:}
As shown in Fig.~\ref{ch3fig:observations}(b), the structural overlap between hyperedges increases as the time difference decreases, indicating that new hyperedges tend to overlap more substantially with recent existing ones than with earlier ones.

%% file: 050proposed.tex
\section{Proposed Method: \methodthree}
\label{ch3sec:proposed}

In this section, we introduce the proposed method, \methodthree, a search-based method for hyperedge prediction.

\subsection{Overview (Alg.~\ref{ch3alg:framework})}\label{sec:method:overview}
An overview of \methodthree is presented in Algorithm~\ref{ch3alg:framework}.
Given an observed hypergraph $H=(V, E)$ (optionally with edge timestamps $t_j$'s), \methodthree predicts new hyperedges by identifying the $k$ high-scoring hyperedge candidates, determined by a scoring function designed based on our observations (see Sect.~\ref{ch3sec:observation}).
Here, $k$ is given as an input in the problem statement (see Sect.~\ref{ch3sec:problem}) and is not a hyperparameter of \methodthree we need to fine-tune.

\methodthree first determines the target number $k_i$ of hyperedges to predict for each size $i$ as $k_i := {\operatorname{round}\left(k \frac{| \{ e \in E : | e | = i \} |}{|E|} \right)}$, to match the original hyperedge size distribution in $E$.
Then, \methodthree performs a depth-first search (DFS) over the space of hyperedge candidates
to identify the top-$k_i$ hyperedges for each size, based on the scoring function. However, the original scoring function is not anti-monotonic (i.e., a subset may have a lower score than its supersets), making naïve pruning strategies ineffective.
To address this challenge, \methodthree employs an anti-monotonic upper bound of the original scoring function, which we derive. This allows \methodthree to efficiently explore the search space with theoretical guarantees.

    
We shall describe two key components of \methodthree: 
(1) scoring based on empirical patterns that accurately identifies promising hyperedges and 
(2) search with an anti-monotonic upper bound that enhances computational efficiency.

\vspace{-2mm}
\subsection{Scoring Based on Observations}\label{ch3sec:scoring}
We design the score for each candidate hyperedge $e^{\prime}$ based on our empirical observations in Sect.~\ref{ch3sec:observation}.

\subsubsection{Prioritizing Candidates with High Overlap}\label{ch3sec:relaxed}\label{ch3sec:overlap}
    \hfill\\
    
    \noindent\fbox{
        \parbox{0.95\columnwidth}{
            \textbf{Key Idea From Obs. 1 (Sect.~\ref{ch3obs:obs1}).}
            Ground-truth new hyperedges are likely to significantly overlap with many observed hyperedges.
            Therefore, in our scoring, we should prioritize candidates with high overlap.
        }
    }
    \vspace{0.5mm}

\smallsection{Relaxed Overlap Count.}
Inspired by frequent itemset mining~\cite{poernomo2009towards}, we consider \textit{relaxed overlap count}, controlled by three relaxation ratios between $0$ and $1$ (larger ratio means more relaxed):
(1) node relaxation ratio $\epsilon_v$,
(2) hyperedge relaxation ratio $\epsilon_e$, and
(3) total relaxation ratio $\epsilon_t$.
Let $E$ be the set of observed hyperedges, for each candidate hyperedge $e'$, we say a subset $\widetilde{E} \subseteq E$ satisfies the criteria for $e'$ w.r.t. the three relaxation ratios, if:
    (1) for each node $v'$ in $e'$, it is missing in at most $\epsilon_v$ proportion of the hyperedges in $\widetilde{E}$, i.e.,
    $\forall v^{\prime} \in e^{\prime}, | \{ e \in \widetilde{E}: v^{\prime} \notin e \} | \leq {\epsilon}_v | \widetilde{E} |$;
    (2) in each hyperedge $e \in \widetilde{E}$, at most $\epsilon_e$ proportion of the nodes in $e'$ are missing, i.e.,
    $\forall e \in \widetilde{E}, |v' \in e' : v' \notin e| \leq \epsilon_e |e'|$;
    and (3) the total missing occurrences of nodes in $e'$ in $\widetilde{E}$ has proportion at most $\epsilon_t$, i.e.,
    $|v' \in e', e \in \widetilde{E} : v' \notin e| \leq \epsilon_t |e'| |\widetilde{E}|$.    

We let ${\widetilde{E}}(e', \epsilon_v, \epsilon_e, \epsilon_t)$ be the maximum subset satisfying the above conditions, and define the \textit{relaxed overlap count} as $ovr(e', \epsilon_v, \epsilon_e, \epsilon_t) = |{\widetilde{E}}(e', \epsilon_v, \epsilon_e, \epsilon_t)|$.~\footnote{Considering such relaxation enhances the flexibility in scoring.
Specifically, when relaxation is disabled, i.e., when $\epsilon_v = \epsilon_e = \epsilon_t = 0$, $ovr(e', 0, 0, 0)$ would be simply the number of observed hyperedges that are supersets of $e'$.}
The set $\widetilde{E}$ consists of hyperedges that collectively overlap with the nodes in $e'$. For example, if most hyperedges overlap only with a subset of $e'$, they do not sufficiently support the likelihood of the entire $e'$ and are therefore excluded from $\widetilde{E}$ due to the first constraint.
Similarly, hyperedges that overlap only marginally with a small portion of $e'$ are excluded due to the second constraint.
If we use $E$ instead of $\widetilde{E}$ for scoring (e.g., for the overlap count), however,
such hyperedges, which do not meaningfully support the likelihood of $e'$, increase its score, degrading the validity of the score.



        The set ${\widetilde{E}}(e', \epsilon_v, \epsilon_e, \epsilon_t)$ can be computed by solving an Integer Linear Programming (ILP) problem.~\footnote{In our implementation, we solve this ILP formulation using the SCIP solver~\cite{achterberg2009scip} provided via Google OR-Tools.}
        Let the binary decision variable $x_j \in \{0,1\}$ for each $e_j \in E$ indicate whether $e_j$ is included in $\widetilde{E}$ ($x_j=1$) or not ($x_j=0$), and let $A \in \{0,1\}^{|e'| \times |E|}$ be a binary matrix where $A_{i,j}$ indicates whether $v_i \notin e_j$ ($A_{i,j}=1$) or not ($A_{i,j}=0$).
        Then, the above constraints are formulated as follows:
        \begin{enumerate}[leftmargin=*]
            \item \textbf{Node relaxation constraint:}
            \begin{equation}
                \sum\nolimits_{j=1}^{|E|} A_{i,j} \cdot x_j \leq \epsilon_v \cdot \sum\nolimits_{j=1}^{|E|} x_j \quad \forall i \in [|e'|].
            \end{equation}
            \item \textbf{Hyperedge relaxation constraint:}
            \begin{equation}
                \sum\nolimits_{i=1}^{|e'|} A_{i,j} \leq \epsilon_e \cdot |e'| \quad \forall j \text{ such that } x_j = 1.
            \end{equation}
            \item \textbf{Total relaxation constraint:}
            \begin{equation}
                \sum\nolimits_{i=1}^{|e'|} \sum\nolimits_{j=1}^{|E|} A_{i,j} \cdot x_j \leq \epsilon_t \cdot |e'| \cdot \sum\nolimits_{j=1}^{|E|} x_j.
            \end{equation}
        \end{enumerate}
        The objective then is to maximize the size of $\widetilde{E}$, i.e.,
        \begin{equation}
            \max \sum\nolimits_{j=1}^{|E|} x_j.
        \end{equation}

    \smallsection{Incorporating Overlap Ratio.}        
    Relaxed overlap count only accounts for the \textit{number} of observed hyperedges that satisfy the criteria, so we further incorporate \textit{overlap ratio} (see Sect.~\ref{ch3obs:obs1}) into the scoring function to capture the degree of overlap of each observed hyperedge.
    The final score is
    \begin{equation}    
    f_1(e') = \sum\nolimits_{e \in \widetilde{E}(e^{\prime}, {\epsilon}_v, {\epsilon}_e, {\epsilon}_t)} {\frac{| e^{\prime} \cap e |}{| e |}}.
    \end{equation}
        
        

\vspace{2mm}
\subsubsection{Weighting More Recent Observed Hyperedges}\label{ch3sec:time}
    \hfill\\
    
    \noindent\fbox{
        \parbox{0.95\columnwidth}{
            \textbf{Key Idea From Obs. 2 (Sect.~\ref{ch3obs:obs3}).}
            Ground-truth new hyperedges are likely to overlap with more recent observed hyperedges than earlier ones.
            Therefore, in the scoring, we should add higher weights to more recent observed hyperedges.            
        }
    }
    \vspace{0.5mm}

    \smallsection{Time weight.}
    We introduce \textit{time weight} that assigns greater significance to more recent hyperedges, when edge timestamps are available.
    Formally, the time weight for an observed hyperedge $e$ is defined as:
    $\exp(\tau t_{e})$,    
    where $\tau$ is an adjustable parameter that determines the emphasis on recent hyperedges, and $t_{e} \in [0, 1]$ is the normalized timestamp of $e$.
    With time weight included, the final score is
    \begin{equation}
        f_2(e') = \sum\nolimits_{e \in \widetilde{E}(e^{\prime}, {\epsilon}_v, {\epsilon}_e, {\epsilon}_t)} {\frac{| e^{\prime} \cap e |}{| e |}} \exp(\tau t_{e}).
    \end{equation}
    

    \normalem
    \begin{algorithm2e}[t!]
    \DontPrintSemicolon
    \small
    \LinesNotNumbered
    \KwIn{
    (1) $H=(V, E)$: observed hypergraph \tcp*{Sect.~\ref{ch3sec:problem}} \\
    \textit{\textcolor{gray}{(optionally with edge timestamps $t_j$'s)}} \\
    (2) $k$: target number of hyperedges to predict \tcp*{Sect.~\ref{ch3sec:problem}}
    (3) (${\epsilon}_v$, ${\epsilon}_e$, ${\epsilon}_t$): relaxation ratios
    \tcp*{Sect.~\ref{ch3sec:overlap}}
    \textit{\textcolor{gray}{(4) $\tau$: growth constant for time weight}}  
    \tcp*{Sect.~\ref{ch3sec:time}}
    }
    \LinesNumbered
    \KwOut{$E^{\prime}$: predicted set of new hyperedges}    
    $i_{max} \gets \max\{ |e| : e \in E \}$ \tcp*{maximum hyperedge size}    
    $k_i \gets {\operatorname{round}\left(k \frac{| \{ e \in E : | e | = i \} |}{|E|} \right)}, \forall i \leq i_{max}$ \tcp*{targets; Sect.~\ref{sec:method:overview}}
    ${\theta}_i \gets 0, \forall i \leq i_{max}$ \tcp*{initialize top-$k$ thresholds}
    $E_{i}^{\prime} \gets \emptyset, \forall i \leq i_{max}$ \tcp*{initialize outputs}
    \ForEach{$v \in V$}{
        $\mathcal{S} \gets [(v, \{v\})]$ \tcp*{initialize DFS; Sect.~\ref{ch3sec:topk}}
        \While{$\mathcal{S} \neq \emptyset$}{
            pop $(u, e^{\prime})$ from $\mathcal{S}$ \tcp*{$e^{\prime}$: current candidate}
            $i \gets |e'|$ \tcp*{$i$: current size}
            \textbf{if} $f_t(e') < \theta_i$ \textbf{then} continue \tcp*{prune $e'$ and all its supersets if its score's upper bound $f_t(e')$ does not exceed the top-$k$ threshold; Sect.~\ref{ch3sec:bound}}
            calculate score $f_s(e')$ \tcp*
            {Sect.~\ref{ch3sec:scoring}}
            \If{$|E'_i| < k_i$ \textbf{or} $f_s(e') \geq \theta_i$}{
                insert $e'$ into $E'_i$ and retain the top-$k_i$ by score\;
                $\theta_i \gets \min_{e \in E'_i} f_s(e)$ \tcp*{update $\theta_i$; Sect.~\ref{ch3sec:topk}}
            }
            \textbf{if} {$|e^{\prime}| = i_{max}$} \textbf{then} continue \tcp*{too large size}            
            \ForEach{$w \in V \setminus e^{\prime}$}{                push $(w, e^{\prime} \cup \{ w \})$ into $\mathcal{S}$ \tcp*{proceed DFS}
            }
        }            
        }
        $E' \gets \bigcup_{i \leq i_{max}} E'_i$ \tcp*{collect outputs}
    \KwRet{$E'$}
    \caption{Overview of \methodthree}\label{ch3alg:framework}
    \end{algorithm2e}
    \ULforem

\subsubsection{Final Scoring Function}\label{ch3sec:finalscoring}
For datasets with edge timestamps, the final score for a hyperedge candidate $e'$ is:
\begin{equation} \label{ch3eq:finalscoring}
        f_s(e^{\prime}) = f_2(e') = \sum\nolimits_{e \in \widetilde{E}(e^{\prime}, {\epsilon}_v, {\epsilon}_e, {\epsilon}_t)} {\frac{| e^{\prime} \cap e |}{| e |} \exp({\tau t_{e}})}.
    \end{equation}
When edge timestamps are not available, $\exp({\tau t_{e}})$ is omitted, i.e., we use $f_s(e') = f_1(e') = \sum_{e \in \widetilde{E}(e^{\prime}, {\epsilon}_v, {\epsilon}_e, {\epsilon}_t)} {\frac{| e^{\prime} \cap e |}{| e |}}$.

\subsection{Pruning with Anti-Monotonic Upper Bound}\label{ch3sec:bound}
As discussed in Sect.~\ref{ch3sec:intro}, with an unconstrained candidate set, the naïve exhaustive enumeration of all $\mathcal{O}(2^{|V|})$ candidates and picking the best candidates is computationally infeasible in most cases.
We aim to effectively and efficiently prune unconstrained candidate sets with theoretical guarantees.

For pruning set functions, a useful property is \textit{anti-monotonicity}.
A set function $f(\cdot)$ is anti-monotonic if $f(e_1) \geq f(e_2)$ for any $e_1 \subseteq e_2$.
If a function is anti-monotonic, when we identify a sub-optimal candidate $e'$, we can prune all its supersets, because all its supersets are never better than itself.
However, our original scoring function is not anti-monotonic in general (i.e., for arbitrary relaxation ratios).
To this end, we aim to derive an \textit{anti-monotonic upper bound} $f_n$ of the original scoring function $f_s$.
The key idea is, once we have $f_n$, we can use it as a \textit{safe pruning criterion}:
if a candidate $e'$ satisfies $f_n(e') < \theta$ for some pruning threshold $\theta$, then all supersets $e''$ of $e'$ can be pruned without loss of optimality:
\begin{equation}\label{eq:explain_anti_ub}
    f_s(e'') \overset{\text{upper bound}}{\leq} f_n(e'') 
    \overset{\text{anti-mono.}}{\leq} f_n(e') < \theta, \forall e'' \supseteq e'.
\end{equation}

Specifically, we leverage the following anti-monotonic upper bound $f_n$ of $f_s$:
\begin{equation}\label{eq:anti_ub}
\small
    f_n(e') = \sum\nolimits_{e \in \widetilde{E}(e^{\prime}, {\epsilon}_v, 1, 1)} {\exp({\tau t_{e}})},
\end{equation}
where we only keep the node relaxation ratio $\epsilon_v$ while making the other two ratios maximally relaxed (i.e., with ratio 1), and 
relax the overlap ratio $\frac{| e^{\prime} \cap e |}{| e |}$ to its upper bound of $1$. 
We omit $\exp({\tau t_{e}})$ when timestamps are unavailable.

\begin{theorem}\label{thm:anti_ub}
    $f_n$ is an anti-monotonic upper bound of $f_s$.
\end{theorem}
\begin{proof}
\textbf{(Anti-monotonicity)} 
By definition, a function $f_n$ is anti-monotonic if $e'' \supseteq e'$ implies $f_n(e'') \leq f_n(e')$.
To show that $f_n$ satisfies this property, it suffices to show that for any $e'' \supseteq e'$, 
$\widetilde{E}(e'', \epsilon_v, 1, 1) \subseteq \widetilde{E}(e', \epsilon_v, 1, 1)$.
Indeed, since $e'' \supseteq e'$, the node-level condition imposed by $e''$ is at least as strict as that of $e'$.
That is, if an observed hyperedge $e$ satisfies 
$|\{ \hat{e} \in \widetilde{E} : v' \notin \hat{e} \}| \leq \epsilon_v \cdot |\widetilde{E}|$ for all $v' \in e''$, then the condition trivially holds for all $v' \in e'$, since $e' \subseteq e''$.
Hence, any hyperedge satisfying the constraint for $e''$ also satisfies it for $e'$, 
and the inclusion of support sets follows.
Since $f_n$ sums over these sets, we conclude $f_n(e') \geq f_n(e'')$.

\textbf{(Upper bound)} 
We aim to show that $f_s(e) \leq f_n(e)$ for any candidate hyperedge $e$.
This holds if $\widetilde{E}(e, \epsilon_v, \epsilon_e, \epsilon_t) \subseteq \widetilde{E}(e, \epsilon_v, 1, 1)$ since the overlap ratio in $f_s(e)$ is relaxed to its upper bound of 1 in $f_n(e)$.
Indeed, the definition of $\widetilde{E}$ ensures that increasing the relaxation parameters $\epsilon_e$ and $\epsilon_t$ (i.e., making the edge-level and time-level conditions less strict) can only expand the set of observed hyperedges that satisfy the similarity constraints.
Therefore, any $e \in \widetilde{E}(e, \epsilon_v, \epsilon_e, \epsilon_t)$ also belongs to $\widetilde{E}(e, \epsilon_v, 1, 1)$, implying $\widetilde{E}(e, \epsilon_v, \epsilon_e, \epsilon_t) \subseteq \widetilde{E}(e, \epsilon_v, 1, 1)$.
\end{proof}

Computing $f_n$ is NP-complete~\cite{poernomo2009towards}, and we further derive an upper bound $f_t$ of $f_n$ that allows efficient evaluation:
\begin{equation}\label{eq:practical_ub}
    f_t(e') = \sum\nolimits_{e \in \widetilde{E}(e^{\prime}, 1, 1, \epsilon_v)} {\exp({\tau t_{e}})},
\end{equation}
where we replace the node relaxation condition with the total relaxation condition with the same ratio.
The key idea is that, when only total relaxation is involved, we can efficiently construct $\widetilde{E}(e^{\prime}, 1, 1, \epsilon_v)$ by greedily including hyperedges $e$ in the ascending order w.r.t the number of missing nodes, i.e., $|v' \in e': v' \notin e|$, until the relaxation ratio $\epsilon_v$ is reached.

\begin{lemma}\label{lem:practical_ub}
    $f_t$ is an upper bound of $f_s$.
\end{lemma}
\begin{proof}
It suffices to show that, $\widetilde{E}(e', {\epsilon}_v, 1, 1) \subseteq \widetilde{E}(e', 1, 1, {\epsilon}_v)$, for any $e'$.
Indeed, for any observed hyperedge $e$ such that $\forall v^{\prime} \in e', | \{ e \in \widetilde{E}: v^{\prime} \notin e \} | \leq {\epsilon}_v | \widetilde{E} |$,
we immediately have
$|v' \in e', e \in \widetilde{E} : v' \notin e| = \sum_{v' \in e'}|\{e \in \widetilde{E} : v' \notin e\}| \leq \sum_{v' \in e'} {\epsilon}_v | \widetilde{E} | =
\epsilon_v |e'| |\widetilde{E}|$.
\end{proof}

Now, $f_t$ can be used instead of $f_n$ for pruning with improved efficiency, while maintaining theoretical guarantees.
If a candidate $e'$ satisfies $f_t(e') < \theta$ for some pruning threshold $\theta$, then all supersets $e''$ of $e'$ can be pruned without loss of optimality
(cf. Eq.~\eqref{eq:anti_ub}):
\begin{equation}\label{eq:explain_practical_ub}
\small
    f_s(e'') \overset{\text{upper bd.}}{\leq} f_n(e'') 
    \overset{\text{anti-mono.}}{\leq} f_n(e') \overset{\text{upper bd.}}{\leq} f_t(e') < \theta.
\end{equation}

In practice, we use $f_t$ to quickly prune sub-optimal candidates, and then use the original scoring function $f_s$ to evaluate the ``surviving'' candidates that pass the filtering of $f_t$.

\subsection{DFS-based Candidate Exploration and Top-$k$ Maintenance}\label{ch3sec:topk}
We use a DFS-based search to explore the whole candidate set, and maintain a list of candidates with the highest scores on the fly.
As mentioned in Sect.~\ref{sec:method:overview}, we first determine the target number $k_i$ of hyperedges to predict for each size $i$, to match the original hyperedge size distribution in $E$.

We use DFS to explore the candidate set, from small candidates to large ones.
Along the search, for each size $i$, we maintain a list $E'_i$ of top-$k_i$ candidates with size $i$ and the corresponding threshold $\theta_i = \min_{e' \in E'_i} f_s(e')$.
For each candidate $e'$ with size $i = |e'|$, we use the efficient upper bound $f_t$ to evaluate it.
If $f_t(e') < \theta_i$, then we can safely skip $e'$ and all its supersets (see Eq.~\eqref{eq:explain_practical_ub}), without loss of optimality.
Otherwise if $e'$ passes the filtering of $f_t$ (i.e., $f_t(e') \geq \theta_i$), we compute the original score $f_s(e')$ and update the top-$k$ list and threshold accordingly.

After checking all candidates, we collect the top-$k_i$ hyperedges from different hyperedge sizes $i$ as the final predictions.
\color{black}

\subsection{Extension with Node Features}\label{ch3sec:nodefeatsim}

In the real world, many hypergraph datasets are accompanied by node features, and the proposed method \methodthree can be easily extended to cases with node features.
Specifically, each node $v_i$ can be associated with a node-feature vector $x_i \in \mathbb{R}^f$ in dimension $f$.
For the real-world datasets used in this work, the nodes in the Co-citation and Authorship datasets are publications (see Sect.~\ref{ch3sec:prelim:data}), and each node is associated with text information (the abstract of the corresponding publication).
We extract one-hot bag-of-word features from such information as node features.

\smallsection{Observations.}
We analyze the similarity between node features within hyperedges, and have the following observations.
For each hyperedge $e$,~\footnote{Specifically, we use each hyperedge $e\in E_{2}$, following the setup in Sect.~\ref{ch3obs:obs1}.
} we compute the Jaccard index between the node features of each node pair in $e$, take the average, and compare it to that in random hyperedges generated as in Sect.~\ref{ch3obs:obs1}.
As shown in Fig.~\ref{ch3fig:obs_nodefeat}, we observe that ground-truth hyperedges exhibit significantly higher node feature similarity between the node pairs than random ones.
This suggests that new hyperedges are more likely to form among nodes with similar features.

\begin{figure}
    \centering        
    \includegraphics[width=0.27\textwidth]{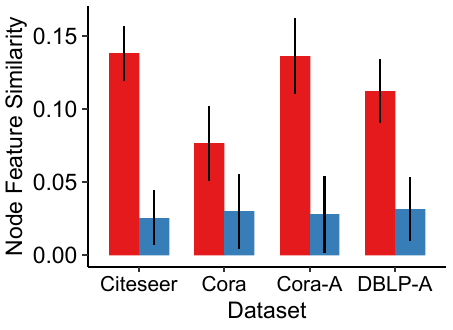}
    \hspace{2pt}
    \raisebox{1.5\height}{\includegraphics[width=0.15\textwidth]{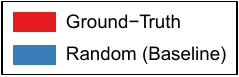}}
    \vspace{-2pt}
    \caption{\label{ch3fig:obs_nodefeat}
        \textbf{\uline{Nodes within a hyperedge have similar features.}}
        Ground-truth hyperedges show higher Jaccard similarity between features of node pairs, suggesting \uline{new hyperedges tend to form between nodes with similar features}.
    }
\end{figure}

\smallsection{Feature Weight.}
Based on our observations above, we introduce \textit{feature weight} that prioritizes candidate hyperedges consisting of nodes with similar features.
For each candidate $e'$, we compute the average Jaccard index 
$X_{e^{\prime}} = \frac{1}{\binom{|e'|}{2}} \sum_{v_i, v_j \in e'} \operatorname{Jaccard}(x_i, x_j)$
between the node pairs in $e'$, together with a tunable hyperparameter $\alpha \geq 0$ to adjust the significance of feature weight.
With feature weight, the final scoring function is $\hat{f}_s(e') = {(X_{e^{\prime}})}^{\alpha} f_s(e')$, and the corresponding upper bound for pruning is
$\hat{f}_t(e') = {(X_{e^{\prime}})}^{\alpha} f_t(e')$.

\smallsection{Discussion.}
A limitation regarding theoretical guarantees, when using node features, is that the average Jaccard index 
$X_{e^{\prime}}$ is not anti-monotonic, and we thus cannot fully guarantee that the pruned candidates are never better than the kept ones.
For theoretical rigor, one may still use ${f}_t$ as the upper bound for pruning, since 
$\hat{f}_s(e'') \leq {f}_s(e'') \leq f_t(e')$, where we have used $(X_{e^{\prime}})^{\alpha} \leq 1$.
However, empirically we observe that using $\hat{f}_t(e')$ shows much better efficiency, which is likely because ${f}_t$ is too loose as an upper bound compared to $\hat{f}_t$.

    

\color{black}


%% file: 060experiment.tex
\section{Experiments}
\label{ch3sec:experiments}

In this section, we present experimental results to evaluate the proposed method \methodthree, show its empirical superiority, and analyze the effectiveness of its key components.
The experiments aim to answer the following research questions:

\begin{itemize}[leftmargin=*]
    \item \textbf{Q1. Accuracy}: How accurately does \methodthree predict new hyperedges, compared to baselines?
    \item \textbf{Q2. Speed and Scalability}: How does the running time of \methodthree compare to baselines, and how does it scale with data sizes and the number of hyperedges to predict?
    \item \textbf{Q3. Ablation Studies}: How does each component of \methodthree contribute meaningfully to its prediction accuracy?
    \item \textbf{Q4. Hyperparameter Sensitivity}: How does the accuracy of \methodthree vary with hyperparameter settings?
\end{itemize}

\subsection{Experimental Settings}\label{ch3sec:settings}

    In this section, we present the dataset, method setups, and evaluation metrics. Detailed hardware and software specifications are provided in Sect.~A of the online appendix~\cite{appendix}.

    \begingroup
\renewcommand{\arraystretch}{0.90}
\setlength{\tabcolsep}{4pt}
    \begin{table*}
        \centering
        \caption{\label{ch3tab:Q1_tab1}
           \textbf{Q1. Accuracy.} \uline{\methodthree performs better than all the baselines consistently across all settings.}           
           We report the average values and standard deviations over the five random splits of Recall$@\mathcal{K}$, across the four datasets without edge timestamps.
           For each setting, the best and second-best methods are highlighted in \green{green} and \yellow{yellow}, respectively.           
        }
        \scalebox{0.92}{
        \begin{tabular}{l|ccc|ccc|ccc|ccc}
        \hline
            {\textbf{Dataset}} & \multicolumn{3}{|c|}{\textbf{Citeseer}} & \multicolumn{3}{|c|}{\textbf{Cora}}  & \multicolumn{3}{|c|}{\textbf{Cora-A}} & \multicolumn{3}{|c}{\textbf{DBLP-A}} \bigstrut \\
            \hline
            \textbf{Method ($\downarrow$) / $\mathbf{\mathcal{K}}$ ($\rightarrow$)} & $1\times$ & $2\times$ & $5\times$ & $1\times$ & $2\times$ & $5\times$ & $1\times$ & $2\times$ & $5\times$ & $1\times$ & $2\times$ & $5\times$ \bigstrut \\ 
            \hline
            \methodthree (Proposed) & \cellcolor{green}8.2 (1.6) & \cellcolor{green}10.9 (1.5 & \cellcolor{green}17.9 (1.8) & \cellcolor{green}7.5 (1.8) & \cellcolor{green}10.0 (2.0) & \cellcolor{green}14.6 (1.5) & \cellcolor{green}7.3 (3.6) & \cellcolor{green}10.9 (2.5) & \cellcolor{green}16.4 (2.9) & \cellcolor{green}5.4 (0.1) & \cellcolor{green}8.4 (0.2) & \cellcolor{green}14.3 (0.4) \bigstrut[t] \\ 
             CNS & 1.5 (0.2) & 3.3 (0.8) & 8.8 (1.4) & 2.9 (2.1) & 5.9 (1.5) & \cellcolor{yellow}12.5 (2.1) & 0.3 (0.2) & 0.6 (0.6) & 2.1 (0.8) & 0.7 (0.2) & 1.2 (0.1) & 2.7 (0.2) \\ 
             HPRA & 0.2 (0.4) & 0.3 (0.4) & 0.8 (0.6) & 0.2 (0.2) & 0.6 (0.5) & 2.3 (1.5) & 0.0 (0.0) & 0.1 (0.2) & 0.1 (0.2) & 0.0 (0.0) & 0.0 (0.0) & 0.1 (0.0) \\ 
             MHP & \cellcolor{yellow}2.8 (1.1) & 4.4 (1.3) & \cellcolor{yellow}8.9 (1.4) & 1.2 (0.9) & 2.4 (1.1) & 6.0 (1.6) & \cellcolor{yellow}0.8 (0.2) & \cellcolor{yellow}1.6 (0.2) & \cellcolor{yellow}6.1 (2.8) & - & - & - \bigstrut[b] \\ 
            \hline
             MHP-C & 2.3 (1.0) & \cellcolor{yellow}5.7 (1.7) & - & 4.2 (1.3) & 8.0 (1.5) & - & 0.4 (0.4) & 1.4 (0.7) & 2.6 (0.5) & - & - & - \bigstrut[t] \\ 
             AHP-C & 2.4 (0.9) & 5.2 (1.2) & - & 4.0 (1.0) & \cellcolor{yellow}8.5 (1.8) & - & 0.4 (0.4) & 0.9 (0.6) & 1.7 (0.7) & - & - & - \\ 
             SAGNN-C & 1.8 (0.6) & 4.3 (1.4) & - & 3.8 (1.7) & 7.5 (2.2) & - & 0.3 (0.3) & 0.7 (0.5) & 1.5 (0.6) & 0.7 (0.1) & 1.2 (0.2) & 2.3 (0.4) \\ 
             NHP-C & 2.3 (0.9) & 5.5 (1.2) & - & \cellcolor{yellow}4.2 (1.3) & 7.4 (1.2) & - & 0.4 (0.3) & 0.9 (0.3) & 2.2 (0.6) & \cellcolor{yellow}0.9 (0.2) & \cellcolor{yellow}1.6 (0.2) & \cellcolor{yellow}3.4 (0.2) \\
             MHP-H & 0.3 (0.4) & 0.7 (0.6) & - & 0.6 (0.5) & 1.9 (1.2) & 3.4 (1.4) & 0.1 (0.1) & 0.1 (0.1) & 0.1 (0.1) & - & - & - \\ 
             AHP-H & 0.0 (0.0) & 0.1 (0.1) & - & 0.5 (0.0) & 1.4 (0.0) & 1.8 (0.0) & 0.0 (0.0) & 0.0 (0.0) & 0.0 (0.0) & - & - & - \\
             SAGNN-H & 0.2 (0.2) & 0.4 (0.3) & - & 0.4 (0.4) & 1.2 (0.8) & 2.1 (1.0) & 0.0 (0.0) & 0.0 (0.0) & 0.0 (0.0) & 0.0 (0.0) & 0.0 (0.0) & - \\ 
             NHP-H & 0.1 (0.2) & 0.3 (0.3) & - & 0.6 (0.5) & 1.9 (1.2) & 3.4 (1.5) & 0.1 (0.2) & 0.1 (0.2) & 0.1 (0.2) & 0.0 (0.0) & 0.0 (0.0) & - \bigstrut[b] \\
        \hline
        \multicolumn{13}{l}{-: out-of-time ($>$ 2 days).}
        \end{tabular}
        }
    \end{table*}
    \endgroup

    \begingroup
    \renewcommand{\arraystretch}{0.90}
    \setlength{\tabcolsep}{6pt}
    \begin{table*}
        \centering
        \vspace{-4mm}
        \caption{\label{ch3tab:Q1_tab2}
           \textbf{Q1. Accuracy.} 
           \uline{\methodthree performs better than all the baselines in most settings.}
           We report the Recall$@\mathcal{K}$ across the six datasets with edge timestamps (no standard deviation since there is only one chronological split).
           For each setting, the best and second-best methods are highlighted in \green{green} and \yellow{yellow}, respectively.
        }
        \scalebox{0.92}{
        \begin{tabular}{l|ccc|ccc|ccc|ccc|ccc|ccc}
        \hline
            {\textbf{Dataset}} & \multicolumn{3}{|c|}{\textbf{Enron}} & \multicolumn{3}{|c|}{\textbf{Eu}} & \multicolumn{3}{|c|}{\textbf{High}} & \multicolumn{3}{|c|}{\textbf{Primary}} & \multicolumn{3}{|c|}{\textbf{Ubuntu}} & \multicolumn{3}{|c}{\textbf{Math-sx}} \bigstrut \\
            \hline
            \textbf{Method ($\downarrow$) / $\mathbf{\mathcal{K}}$ ($\rightarrow$)}  & $1\times$ & $2\times$ & $5\times$ & $1\times$ & $2\times$ & $5\times$ & $1\times$ & $2\times$ & $5\times$ & $1\times$ & $2\times$ & $5\times$ & $1\times$ & $2\times$ & $5\times$ & $1\times$ & $2\times$ & $5\times$ \bigstrut \\ 
            \hline
            \methodthree (Proposed) & \cellcolor{green}16.1 & \cellcolor{green}25.6 & \cellcolor{green}33.1 & \cellcolor{green}12.4 & \cellcolor{green}17.3 & \cellcolor{green}26.8 & \cellcolor{green}14.8 & \cellcolor{green}18.3 & 27.3 & \cellcolor{green}7.3 & \cellcolor{green}11.8 & 20.8 & \cellcolor{green}12.0 & \cellcolor{green}15.4 & \cellcolor{green}20.6 & \cellcolor{green}12.1 & \cellcolor{green}17.3 & \cellcolor{green}24.5 \bigstrut[t] \\ 
            CNS & 10.3 & 16.4 & 29.7 & 5.1 & 10.9 & 22.0 & \cellcolor{yellow}12.6 & 13.8 & 18.1 & 4.5 & 7.3 & 11.9 & 1.6 & 2.9 & 6.7 & 3.4 & 6.0 & 11.7 \\ 
            HPRA & 1.7 & 5.8 & 9.2 & 3.5 & 5.8 & 10.2 & 9.0 & 14.8 & 28.1 & 4.7 & 8.1 & 20.4 & 1.1 & 2.0 & 4.4 & 2.2 & 3.9 & 8.1 \\ 
            MHP & 0.3 & 0.6 & 3.6 & 0.3 & 1.0 & 3.4 & 0.9 & 2.9 & 7.4 & 4.3 & 7.7 & 21.1 & - & - & - & - & - & - \bigstrut[b] \\ 
            \hline
            MHP-C & 7.6 & 14.9 & 22.0 & 7.4 & 14.1 & 22.7 & 4.3 & 5.7 & 8.3 & 4.1 & 5.7 & 9.6 & - & - & - & - & - & - \bigstrut[t] \\ 
            SAGNN-C & 6.0 & 8.2 & 14.6 & \cellcolor{yellow}8.6 & \cellcolor{yellow}16.0 & \cellcolor{yellow}24.5 & 7.2 & 8.5 & 9.9 & 5.1 & 7.9 & 11.9 & \cellcolor{yellow}2.3 & \cellcolor{yellow}4.3 & \cellcolor{yellow}8.7 & \cellcolor{yellow}4.7 & \cellcolor{yellow}7.9 & \cellcolor{yellow}14.5 \\ 
            NHP-C & \cellcolor{yellow}13.3 & \cellcolor{yellow}20.3 & \cellcolor{yellow}29.8 & 8.1 & 14.9 & 23.3 & 9.7 & 11.6 & 13.7 & 5.3 & 8.5 & 13.1 & 1.8 & 3.5 & 7.1 & 3.6 & 6.1 & 11.2 \\ 
            MHP-H & 4.6 & 5.4 & 9.5 & 4.4 & 6.5 & 14.2 & 9.1 & 16.3 & 31.8 & 5.4 & 11.4 & \cellcolor{green}22.6 & - & - & - & - & - & - \\ 
            SAGNN-H & 2.8 & 3.4 & 7.1 & 5.4 & 8.1 & 17.8 & 10.2 & \cellcolor{yellow}18.1 & \cellcolor{green}34.1 & 4.5 & 9.4 & 19.5 & 2.0 & 3.5 & - & 3.7 & 6.5 & - \\ 
            NHP-H & 4.5 & 5.2 & 9.5 & 4.6 & 6.7 & 14.2 & 12.4 & 16.6 & \cellcolor{yellow}32.7 & \cellcolor{yellow}6.0 & \cellcolor{yellow}11.7 & \cellcolor{yellow}22.5 & 1.6 & 2.8 & - & 2.6 & 4.6 & - \bigstrut[b] \\ 
        \hline
        \multicolumn{19}{l}{-: out-of-time ($>$ 2 days).}
        \end{tabular}
        }
    \end{table*}
    \endgroup

    \subsubsection{Datasets}
        We used 10 real-world hypergraphs across five domains, as introduced in Sect.~\ref{ch3sec:prelim:data} (see Tab.~\ref{tab:datasets}).
        For datasets with edge timestamps (i.e., the email, contact, and tags datasets), we split the datasets chronologically.
        Specifically, the earliest 80\% of hyperedges were used as observed (training) hyperedges $E_1$, while the most recent 20\% hyperedges were used as new (test) hyperedges $E_2$ we aim to predict.
        For the other datasets (i.e., the co-citation and authorship datasets), hyperedges were randomly split into 80\% observed (training) hyperedges $E_1$ and 20\% new (test) hyperedges $E_2$, using five different random seeds.
        {For hyperparameter tuning, (the most recent or random) 20\% of training hyperedges (i.e., 16\% of the total) were set aside for validation.}

        
            
        For the splits, we make sure that the same hyperedge does not appear in both $E_1$ and $E_2$ (which is possible when repeated hyperedges exist).
        We also make sure that every node in $E_2$ also appears in $E_1$, i.e., there are no unseen nodes in $E_2$.
        The Cora, High, Primary, Ubuntu, and Math datasets contain only hyperedges of size up to 5.
        For computational efficiency, in the other datasets, only hyperedges up to size 10 were retained before splitting into $E_1$ and $E_2$.
        Additionally, hyperedges with rarely occurring hyperedge sizes ($<1$\%) were removed.
        Details are provided in Sect.~B of the online appendix~\cite{appendix}.

    \subsubsection{Method Setups}\label{sec:exp:method_setups}
        All hyperparameters of \methodthree (${\epsilon}_v$, ${\epsilon}_e$, ${\epsilon}_t$, $\tau$, and $\alpha$) were tuned based on validation performance for each dataset,
        and we used \methodthree with the fine-tuned hyperparameters to predict new hyperedges in testing.
        Specifically, we tuned the hyperparameters through a grid search based on the performance w.r.t. Recall$@1\times$ (refer to Sect.~\ref{sec:exp:setting:eval}) on the validation set.
        We considered 27 combinations of the relaxation ratios (${\epsilon}_v$, ${\epsilon}_e$, ${\epsilon}_t$) where each ratio was set to one of $\{(\frac{1}{3}, \frac{1}{4}, \frac{1}{5})\}$ and additionally ${\epsilon}_v= {\epsilon}_e={\epsilon}_t=0.$
        In addition, the hyperparameters related to node feature similarity ($\alpha$) and time weighting ($\tau$) were each selected from the candidate set $\{0, 0.1, 1, 10\}$. 

        We evaluated competing methods, including 
        (1) clique negative sampling (CNS), the strongest sampling-based method from~\cite{patil2020negative}, 
        (2) HPRA~\cite{kumar2020hpra}, and 
        (3) MHP~\cite{yu2024mhp}.         
        The number of predicted hyperedges was set to $\{1, 2, 5\} \times |E_2|$ (see also Sect.~\ref{sec:exp:setting:eval} for evaluation metrics).
        \begin{itemize}[leftmargin=*]
            \item \textbf{CNS}~\cite{patil2020negative}: It generates hyperedges by 
            (1) randomly selecting an observed hyperedge $e$,
            (2) randomly choosing a node in $v \in e$, and
            (3) replacing $v$ with another node that has common neighbors with all the other nodes in $e$.            
            \item \textbf{HPRA}~\cite{kumar2020hpra}: It generates hyperedge by
            (1) randomly selecting an initial node $v_0$,
            (2) iteratively adding nodes with high similarity scores with the already added nodes, based on resource allocation principles (e.g., based on the number of common neighbors).
            \item \textbf{MHP}~\cite{yu2024mhp}: It generates hyperedges by 
            (1) starting from an empty query set, 
            (2) selecting the first node via Preferential Attachment (i.e., high-degree nodes are likely to be chosen), and 
            (3) iteratively adding nodes based on probabilities predicted by a trained model.
        \end{itemize}
        Note that CNS and HPRA have no hyperparameters to tune. S
        ee Sect.~\ref{ch3sec:related} for more details and discussions of these baselines.

        For the baselines that do not use Hypergraph Neural Networks (HNNs), i.e., CNS and HPRA, we further extended them by applying HNNs as a post-processing step to refine hyperedge candidates generated by them.
        We used different HNN models: MHP~\cite{yu2024mhp}, AHP~\cite{hwang2022ahp}, Hyper-SAGNN~\cite{zhang2019hyper}, and NHP~\cite{yadati2020nhp}.
        For each HNN model, the suffix ``-C'' was added when the model was applied to CNS, and the suffix ``-H'' was added when the model was applied to HPRA.
        {For example, MHP-C refers to the approach where MHP is applied to refine hyperedge candidates generated by CNS, while AHP-H indicates that AHP is used to filter hyperedges predicted by HPRA.}
        {In these extended two-stage models, all HNN models were initially trained for the hyperedge prediction task using their respective loss functions in their original implementation.
        In Sect.~C of the online appendix~\cite{appendix}, we provide a summarizes the hyperparameter search spaces used for the HNN models employed in our experiments.
        
        After training, they were applied as a post-processing step to rank and refine the hyperedge candidates generated in the first stage.}
        Specifically, we first used CNS or HPRA to generate twice the target number (i.e., $2k$) of predicted hyperedges, and then used an HNN model to select the top-$k$ ones.

    \subsubsection{Evaluation Metrics}\label{sec:exp:setting:eval}
        To evaluate prediction performance, we used Recall$@\mathcal{K}$, a widely adopted metric in (hyper)edge prediction tasks.
        Recall$@\mathcal{K}$ measures the proportion of ground-truth new hyperedges that are correctly included in the predicted set when $\mathcal{K}$ hyperedges are predicted.
        In this study, the parameter $\mathcal{K}$ is set to represent a multiple of the number of hyperedges in $E_2$.
        For instance, if $\mathcal{K}=2 \times$, the number of predicted hyperedges is $2 \times |E_2|$.

        In addition to Recall$@\mathcal{K}$, we also evaluated prediction accuracy using the average F1 score~\cite{yang2013overlapping}, which considers the similarity between predicted and ground-truth hyperedges.
        The formal definition and the experimental results using the metric are provided in Sect.~D of the online appendix~\cite{appendix}.

\subsection{Q1. Accuracy (Tables~\ref{ch3tab:Q1_tab1} \& \ref{ch3tab:Q1_tab2})}\label{sec:exp:accuracy}
    We evaluated the accuracy of  \methodthree and baselines, and the results are presented in Tables~\ref{ch3tab:Q1_tab1} and \ref{ch3tab:Q1_tab2}.
    
    Notably, \methodthree consistently outperformed all baselines, including deep learning-based models and rule-based approaches (e.g., CNS and HPRA), across most datasets and experimental settings. The only exceptions were in Recall$@5\times$ on the Contact domain datasets (High and Primary), where \methodthree did not achieve the highest performance. Based on our analysis, this is likely because for those datasets, not enough promising (e.g., those with high overlap with the observed ones) candidates are available. A detailed analysis of these cases is provided in Sect.~E of the online appendix~\cite{appendix}.
        
     
      It is also worth noting that  HPRA exhibited scalability limitations, particularly in datasets where the observed hyperedges are divided into multiple connected components (see Sect.~\ref{ch3sec:intro}), performing worse than even simple negative sampling methods (e.g., CNS). In contrast, \methodthree guaranteed flexible exploration of all promising candidates, regardless of the presence of disconnected components. 

      As shown in Sect.~D of the online appendix~\cite{appendix}, the results were consistent in terms of average F1 score.
      
    \smallsection{Case Studies.}
        To qualitatively evaluate the semantic coherence of nodes within predicted hyperedges, we conducted case studies on the two tag datasets, Math and Ubuntu, where the identities of the nodes (tags) are available.
        For each hyperedge size from 2 to 5, we selected the top-scoring predicted hyperedge generated by \methodthree, which are:
        \begin{itemize}[leftmargin=*]
            \item \textbf{Math examples}: [ring-theory, noetherian], [matrices, vectors, vector-spaces], [group-theory, finite-groups, field-theory, abstract-algebra], [calculus, sequences-and-series, real-analysis, integration, convergence]
            \item \textbf{Ubuntu examples}:
                {[drivers, xorg]},
                {[boot, grub2, btrfs]},
                {[dual-boot, boot, live-usb, grub2]},
                {[partitioning, grub2, 16.04, dual-boot, boot]}
        \end{itemize}
        The results show that \methodthree effectively identifies groups (hyperedges) of semantically related tags (nodes).

    \begin{figure*}[t!]
        \centering
        \subfloat[Without edge timestamps]{
            \includegraphics[width=0.334\textwidth]{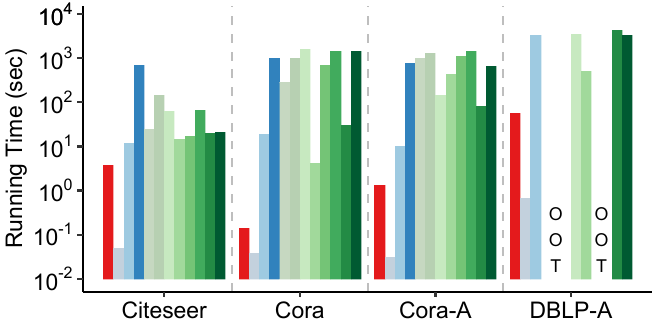}
        }
        \subfloat[With edge timestamps]{
            \includegraphics[width=0.376\textwidth]{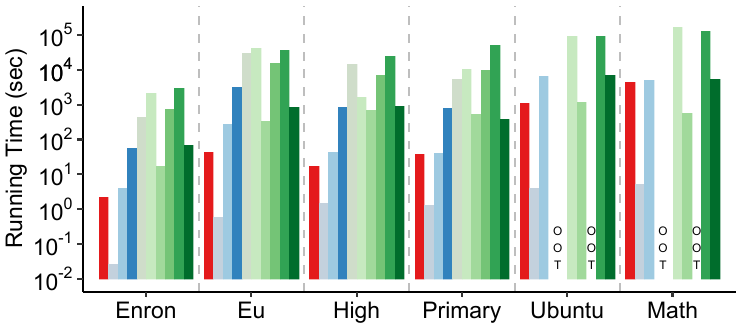}
        }
        \subfloat{
            \raisebox{.6\height}{\includegraphics[width=0.26\textwidth]{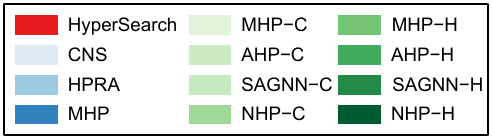}}
        }
    
        \caption{
            \textbf{Q2. Runtime comparison.}
            \uline{\methodthree achieves lower runtime than deep learning-based methods (shown in green) in most cases.}
            For each dataset, we report the running time of \methodthree and baselines.
        }
        \label{ch3fig:time}
        \vspace{-1mm}
    \end{figure*}

\begin{figure*}[t]
        \vspace{-3mm}
        \centering
        \subfloat[\textbf{Relaxed overlap count (Sect.~\ref{ch3sec:relaxed})}\label{ch3fig:ablation1}]{
            \includegraphics[width=0.38\textwidth]{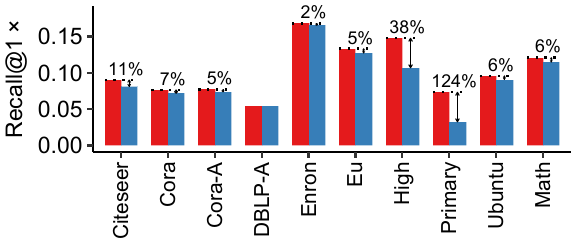}
        }
        \subfloat[\textbf{Time weight (Sect.~\ref{ch3sec:time})}\label{ch3fig:ablation3}]{
            \includegraphics[width=0.255\textwidth]{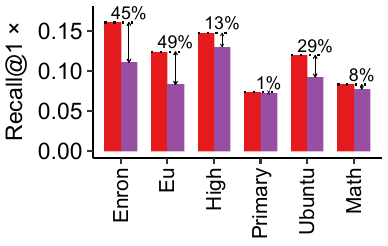}
        }
        \subfloat[\textbf{Feature weight (Sect.~\ref{ch3sec:nodefeatsim})}\label{ch3fig:ablation2}]{
            \includegraphics[width=0.195\textwidth]{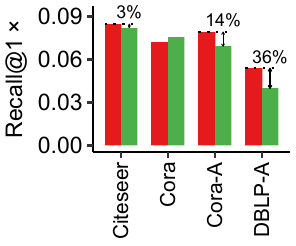}
        }
        \subfloat{
            \raisebox{.77\height}{\includegraphics[width=0.132\textwidth]{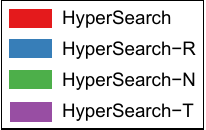}}
        }
        \caption{\label{ch3fig:ablation}
            \textbf{Q3. Ablation studies.}
            \uline{Each component of \methodthree meaningfully contributes to its predictive performance.}
            \methodthree consistently achieves superior or comparable accuracy compared to its variants with some component missing.
        }
    \end{figure*}
    
\begin{figure}[t]
        \centering
        \includegraphics[width=0.6\linewidth]{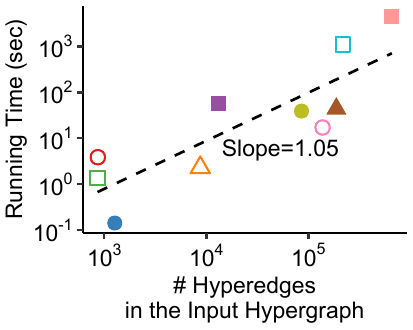}
        \hspace{2pt}
        \raisebox{.33\height}{\includegraphics[width=0.32\linewidth]{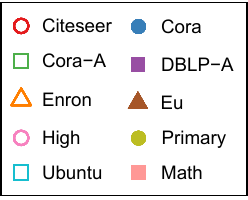}}
        \caption{\label{ch3fig:time_edge}
        \textbf{Q2. Scalability of \methodthree.}
        \uline{The runtime of \methodthree scales nearly linearly with the input hypergraph size, with a regression slope of 1.05, indicating strong scalability.}
        In the plot, we report the number of hyperedges and the running time of \methodthree on each dataset
        }
        \vspace{-1mm}
    \end{figure}

\subsection{Q2. Speed and Scalability (Figs.~\ref{ch3fig:time} \& \ref{ch3fig:time_edge})}\label{sec:exp:speed}
    We evaluated the runtime of the proposed method \methodthree and baselines, and analyzed how \methodthree scales with the size of the input hypergraph.
    The results, presented in Fig.~\ref{ch3fig:time} and Fig.~\ref{ch3fig:time_edge}, reveal the following key observations:
    \begin{itemize}[leftmargin=*]
        \item \methodthree runs faster than deep learning-based methods (highlighted in green in Fig.~\ref{ch3fig:time}) in most cases, demonstrating its superior computational efficiency across diverse datasets.
        \item The runtime of \methodthree scales almost linearly with the number of hyperedges in the input hypergraph, with a regression slope of 1.05 (Fig.~\ref{ch3fig:time_edge}).  
        This suggests that \methodthree maintains scalability as the data size increases.
    \end{itemize}
    In Sect.~G of the online appendix~\cite{appendix}, we show that this scalability persisted on larger synthetic hypergraphs, generated by replicating each original hypergraph up to five times.

\subsection{Q3. Ablation Studies (Fig.~\ref{ch3fig:ablation})}\label{sec:exp:ablation}
    We conducted ablation studies on
    (1) relaxed overlap count in Sect.~\ref{ch3sec:relaxed},
    (2) time weight in Sect.~\ref{ch3sec:time}, and
    (3) feature weight in Sect.~\ref{ch3sec:nodefeatsim}
    to assess their impact on performance.
    Specifically, we compared Recall$@1\times$ between \methodthree and its variants without each component:
    (1) \methodthree-R without relaxation (i.e., $\epsilon_v = \epsilon_e = \epsilon_t = 0$),
    (2) \methodthree-T without time weight (i.e., $\tau = 0$) , and
    (3) \methodthree-N without feature weight (i.e., $\alpha = 0$).
    As shown in Fig.~\ref{ch3fig:ablation}, \methodthree consistently achieves superior or comparable performance compared to its variants, demonstrating that each component meaningfully contributes to its prediction accuracy.

\vspace{-2mm}
\subsection{Q4. Hyperparameter Sensitivity (Fig.~\ref{ch3fig:sensitivity})}
    We analyzed the sensitivity of all hyperparameters in \methodthree, including the relaxation ratios (${\epsilon}_v$, ${\epsilon}_e$, ${\epsilon}_t$), the feature weight ($\alpha$), and the time weight ($\tau$).
    Specifically, we evaluated Recall$@1\times$ by varying (${\epsilon}_v$, ${\epsilon}_e$, ${\epsilon}_t$) across $28$ combinations described in Sect.~\ref{sec:exp:method_setups}
    and by varying $\alpha$ and $\tau$ in $\{0, 0.1, 1, 10\}$, for each dataset.
    
    As shown in Figs.~\ref{ch3fig:sensitivity}(a) and (b), the optimal relaxation ratios (${\epsilon}_v$, ${\epsilon}_e$, ${\epsilon}_t$) vary across datasets.
    Fig.~\ref{ch3fig:sensitivity}(c) shows that, in most cases, \methodthree achieves peak accuracy at specific $\alpha \in \{0.1, 1\}$ for most datasets.
    On the other hand, as shown in Fig.~\ref{ch3fig:sensitivity}(d), increasing $\tau \in \{0, 0.1, 1, 10\}$ monotonically improves the accuracy of \methodthree in 5 out of 6 datasets.
    These results suggest that the impact of hyperparameters depends on dataset characteristics, as different datasets benefit from varying degrees of relaxation, time weight, and feature weight.
    Notably, in our experiments, via validation, we found hyperparameters that gave competitive performance, making \methodthree outperform baselines in most cases.    

    \begin{figure}[t]
        \centering
    
        \begin{minipage}[c]{.49\linewidth}
            \centering
            \includegraphics[width=0.77\textwidth]{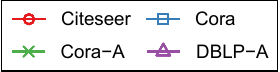}
        \end{minipage}
        \begin{minipage}[c]{.49\linewidth}
            \centering
            \includegraphics[width=0.71\textwidth]{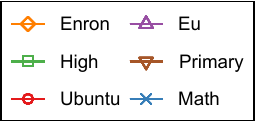}
        \end{minipage}
        \subfloat[Relaxation ratio $\epsilon$'s\\(Without edge timestamps)\label{ch3fig:sensitivity_fault1}]{
            \includegraphics[width=0.48\linewidth]{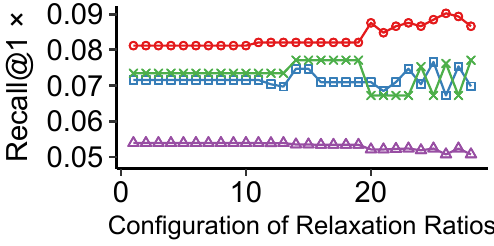}
        }
        \subfloat[Relaxation ratio $\epsilon$'s\\(With edge timestamps)\label{ch3fig:sensitivity_fault2}]{
            \includegraphics[width=0.48\linewidth]{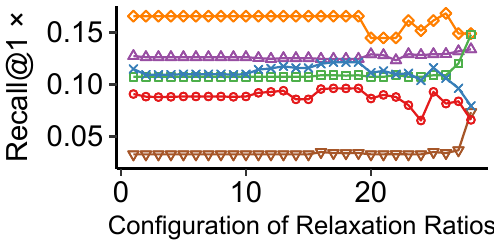}
        }
        \vspace{5pt}
        
        \includegraphics[width=0.47\linewidth]{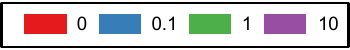}
    
        \vspace{-5pt}
        \subfloat[Feature coefficient $\alpha$\label{ch3fig:sensitivity_node}]{
            \includegraphics[width=0.421\linewidth]{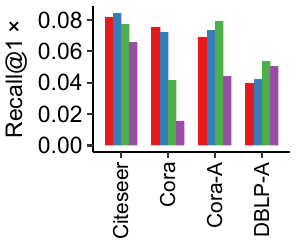}
        }
        \subfloat[Time coefficient $\tau$\label{ch3fig:sensitivity_time}]{
            \includegraphics[width=0.549\linewidth]{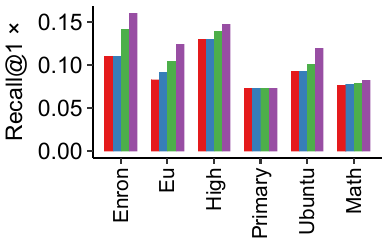}
        }
    
        \caption{
            \textbf{Q4. Hyperparameter sensitivity.}
            \uline{The best hyperparameter choices vary from dataset to dataset.}
            The trends of the accuracy of \methodthree across different hyperparameters (${\epsilon}_v$, ${\epsilon}_e$, ${\epsilon}_t$), $\alpha$, and $\tau$.
        }
        \label{ch3fig:sensitivity}
    \end{figure}

    \subsection{Additional Experiments}
    Due to the page limit, we provide extra experimental results
    in the online appendix~\cite{appendix},
    summarized as follows:
    \begin{itemize}[leftmargin=*]
        \item \textbf{Combination with neural networks (Sect.~F):}  We tried applying a filtering step using hypergraph neural networks to refine the predictions of \methodthree, similar to what we have done for CNS and HPRA (see Sect.~\ref{sec:exp:method_setups}).
        However, this combination did not consistently improve the performance of \methodthree, indicating that its score function based on the patterns in real-world hypergraphs (see Sect.~\ref{ch3sec:observation}) is already effective.
        \item \textbf{Scalability w.r.t. the number of predictions (Sect.~H):} The running time of \methodthree scaled nearly linearly with the number of predictions it produced.
        \item \textbf{Failure analysis (Sect.~E):} We investigated potential reasons for the relatively poor performance of \methodthree on the Contact domain datasets.
        \item \textbf{Additional metric (Sect.~D) and large-scale synthetic datasets (Sect.~G)}: We confirmed that our results remained consistent in terms of average F1 score, and across large-scale synthetic datasets created by replicating hypergraphs.
    \end{itemize}

%% file: 070conclusion.tex
\section{{Conclusion and Discussion}}
\label{ch3sec:conclusion}
In this work, we studied the problem of hyperedge prediction.
We identified two key limitations of existing methods (Sect.~\ref{ch3sec:intro}).
We analyzed and discussed several observations in real-world hypergraphs on structure and timestamps (Sect.~\ref{ch3sec:observation}).
We proposed a novel method, \methodthree, for hyperedge prediction, consisting of two key components: (1) a scoring function based on our observations to evaluate candidate hyperedges, and (2) an efficient search scheme using an anti-monotonic upper bound of the scoring function (Sect.~\ref{ch3sec:proposed}).
We conducted extensive experiments on 10 real-world hypergraphs across five domains to validate the empirical superiority of \methodthree regarding both accuracy and efficiency (Sect.~\ref{ch3sec:experiments}).

\smallsection{Discussion on Limitations.}
While \methodthree demonstrates strong empirical efficiency, its theoretical complexity is high due to reliance on integer programming~\cite{poernomo2009towards}. 
Moreover, for hypergraphs with features, incorporating the average Jaccard index compromises the anti-monotonicity, as discussed in Sect.~\ref{ch3sec:nodefeatsim}, where we discussed a workaround for theoretical rigor (i.e., use $1$ as a trivial upper bound of the average Jaccard index).
We would like to explore theoretically sound and empirically effective upper bounds when additional information (e.g., node features) is available, as a future direction.

\smallsection{Reproducibility.}
Our code, datasets, and appendix 
are publicly available at \url{https://github.com/jin-choo/HyperSearch/}.